\documentclass[manyauthors]{fundam}

\publyear{22}
\papernumber{2102}
\volume{185}
\issue{1}

\usepackage{url} 
\usepackage[ruled,lined]{algorithm2e}
\usepackage{graphicx}
\usepackage{tikz}
\usetikzlibrary{automata, positioning, arrows, calc}

\usepackage{amsmath, amssymb, dsfont}
\usepackage{mathtools}
\usepackage{hyperref} 
\usepackage[breakable,skins,xparse]{tcolorbox} 
\usepackage{booktabs} 

\ExplSyntaxOn

\colorlet{commentcolor}{gray} 

\DeclareTColorBox {CommentBox} {m}
{
	code         = {\linespread{.9}},
	frame~hidden,
	boxrule      = 0pt,
	breakable,
	enhanced,
	before~skip  = 10pt plus 4pt minus 4pt,
	toptitle     = 4pt,
	boxsep       = 3pt,
	left         = 6pt,
	right        = 6pt,
	top          = 2pt,
	bottom       = 3pt,
	sharp~corners,
	colback 	 = commentcolor,
	colbacktitle = commentcolor,
	title        = {#1},
	fonttitle    = \bfseries\small,
	fontupper    = \small,
	parbox       = false,
	coltitle     = black,
	after~skip   = 6pt plus 4pt minus 4pt,
	opacityback  = 0.2,
	opacitybacktitle=0.2
}

\ExplSyntaxOff

\newcommand\sdots{\makebox[1em][c]{.\hfil.\hfil.}}
\newcommand{\einsum}[4][\#]{#1\left(\begin{aligned}#2 \rightarrow #3;\;#4\end{aligned}\right)}

\def\domain{\text{dom}}
\newcommand{\tensor}[3][]{%
	\ifthenelse{\equal{#1}{}}{%
		\ifthenelse{\equal{#3}{}}{%
			#2%
		}{%
			#2(#3)%
		}%
	}{%
		\ifthenelse{\equal{#3}{}}{%
			#2_{#1}%
		}{%
			#2_{#1}(#3)%
		}%
	}%
}

\def\N{\mathbb{N}}
\def\R{\mathbb{R}}
\DeclareMathAlphabet{\mymathbb}{U}{BOONDOX-ds}{m}{n}

\begin{document}

\title{The Syntax and Semantics of \emph{einsum}}

\author{Maurice Wenig\corresponding\\
	Institute of Computer Science\\
	Friedrich-Schiller-University Jena \\
	Jena, Germany \\
	maurice.wenig@uni-jena.de
	\and Paul G.\ Rump\\
	Institute of Computer Science\\
	Friedrich-Schiller-University Jena \\
	Jena, Germany \\
	paul.gerhardt.rump@uni-jena.de
	\and Mark Blacher\\
	Institute of Computer Science\\
	Friedrich-Schiller-University Jena \\
	Jena, Germany \\
	mark.blacher@uni-jena.de
	\and Joachim Giesen\\
	Institute of Computer Science\\
	Friedrich-Schiller-University Jena \\
	Jena, Germany \\
	joachim.giesen@uni-jena.de
}

\maketitle

\runninghead{M.\ Wenig, P.\ G.\ Rump, M.\ Blacher, J.\ Giesen}{Syntax and Semantics of \emph{Einsum}}

\begin{abstract} 
	In 2011, \emph{einsum} was introduced to NumPy as a practical and convenient notation for tensor expressions in machine learning, quantum circuit simulation, and other fields.
	It has since been implemented in additional Python frameworks such as PyTorch and TensorFlow, 
	as well as in other programming languages such as Julia. 
	Despite its practical success, the \emph{einsum} notation still lacks a solid theoretical basis, and is not unified across the different frameworks, limiting opportunities for formal reasoning and systematic optimization.
	In this work, we discuss the terminology of tensor expressions
	and provide a formal definition of the \emph{einsum} language.
	Based on this definition, we formalize and prove important equivalence rules for tensor expressions
	and highlight their relevance in practical applications.
\end{abstract}


\clearpage

\section{Introduction}

Historically, the most widespread language for describing machine learning operations has been linear algebra.
Linear algebra captures scalars, vectors, and matrices, but does not support operations on higher-order tensors, which are becoming increasingly important in modern machine learning~\cite{TensorComprehensions}.
Moreover, linear algebra assigns each possible operation its own symbol, which becomes impractical when extended to tensors of arbitrary order.
In the \emph{einsum} notation, tensor expressions are written as an operation on index strings, making it a viable option as a machine modeling language.
Through implementations in computational frameworks such as NumPy \cite{Harris2020}, PyTorch \cite{Paszke2019}, and TensorFlow \cite{TensorFlow}, \emph{einsum} has already proven itself a practical and convenient notation for general tensor expressions.
However, it is not consistently defined across the different frameworks,
which restricts the use of \emph{einsum} as a common interface and limits opportunities for systematic optimization.
Additionally, the \emph{einsum} notation lacks a solid theoretical basis,
precluding formal reasoning about and mathematical proofs of its algebraic properties.

One such property is the fact that semantically equivalent \emph{einsum} expressions often have many different syntactic representations.
This is highly relevant in practice because the syntactic representation of a tensor expression can strongly impact its suitability for a given task.
Often, one representation may be best suited for efficient evaluation, a second for automatic differentiation, a third for the deduction of relevant properties such as symmetry or convexity,
and yet another for human readability.
For example, the matrix-matrix-vector product $A \cdot B \cdot v$ can be written as:
\begin{itemize}
	\item \texttt{einsum}$\left(\begin{aligned}ij, jk, k \rightarrow i; A, B, v\end{aligned}\right)$,
	      which is concise and thus human-readable,
	\item \texttt{einsum}$\left(\begin{aligned}ij, j\rightarrow i; A,
				      \texttt{einsum}\left(\begin{aligned}jk, k\rightarrow j; B, v \end{aligned}\right)
			      \end{aligned}\right)$, which specifies an efficient evaluation order that decomposes the operation into two matrix-vector products,	completely avoiding the costly matrix-matrix product $A\cdot B$, and
	\item \texttt{einsum}$\left(\begin{aligned}i, ij, jk, k \rightarrow i; 1, A, B, v\end{aligned}\right)$,
	      which is suitable for differentiation with respect to the matrix $A$, resulting in the derivative expression
	      \texttt{einsum}$\left(\begin{aligned}i, jk, k \rightarrow iij; 1, B, v\end{aligned}\right)$.
	      Without the added ones-vector, all information about the $i$-axis would be lost.
\end{itemize}
So far, however, there are no formal proofs for \textbf{equivalence rules} by which \emph{einsum} expressions can be reshaped.
Therefore, we present a unified definition of the syntax and semantics of \emph{einsum} in this work.
We formally prove the commutativity, associativity, and distributivity of \emph{einsum}, as well as nesting and denesting rules, and several other semantic equivalence rules.
This addresses the theoretical underpinnings of \emph{einsum} and thus complements existing work on its computational aspects~\cite{Breuer2025,Staudt2024,Staudt25,Cooper2020,Hong2022}.

This article is structured as follows:
Because the tensor terminology that is used in the literature is not unified, we start by introducing the terminology that we use throughout the paper in the next section, \hyperref[sec:terminology]{Section~\ref*{sec:terminology}}.
In \hyperref[sec:syntax]{Section~\ref*{sec:syntax}}, we propose a syntax for the \emph{einsum} language and explain reasons why we deviate from existing implementations of \emph{einsum} in certain details.
In \hyperref[sec:semantics]{Section~\ref*{sec:semantics}}, we define the semantics of \emph{einsum},
which are then used in \hyperref[sec:properties]{Section~\ref*{sec:properties}}, \hyperref[sec:nesting]{Section~\ref*{sec:nesting}}, and \hyperref[sec:auxiliary]{Section~\ref*{sec:auxiliary}}
to prove several properties that are important to efficiently evaluate and automatically differentiate \emph{einsum} expressions.
The article is then concluded in \hyperref[sec:conclusions]{Section~\ref*{sec:conclusions}}.

\section{Terminology}
\label{sec:terminology}

In this section, we clarify the fundamental terminology of tensors and their components.
Formally, we define a tensor as a mapping from positions (multi-indices) to values.

\begin{definition}[Tensors]
	Given $o\in\mathbb{N}_0$ and $d_1, \ldots, d_o \in\mathbb{N}$ as well as a set $R$, let $\mathcal{X} = [d_1] \times\ldots\times [d_o]$.
	A \textbf{tensor} is a mapping $T:\mathcal{X} \to R$ of \textbf{positions} $(x_1, \ldots, x_o)\in\mathcal{X}$ to \textbf{entries} $\tensor{T}{x_1,\ldots,x_o}\in R$.
	Here, $\mathbb{N}_0$ denotes the set of natural numbers including zero, $\mathbb{N}$ denotes the set of natural numbers without zero, and $[d]$ denotes the set $\{1,\ldots,d\}$.

	We call $o$ the \textbf{order} of $T$ and say that $T$ has $o$ \textbf{axes} with respective \textbf{axis lengths} $d_1, \ldots, d_o$.
	In the literature, an alternative term for axis is `mode'.
\end{definition}

\begin{CommentBox}{Tensors and data structures}
	Tensors are frequently defined as multidimensional arrays.
	In contrast, our more abstract functional definition remains agnostic to any such specific data structure,
	and therefore also applies to lists (typically used to represent sparse tensors),
	Tensor-Network-Based-Decision-Diagrams (TDDs; useful when representing highly structured tensors)~\cite{Hong2022},
	or any other data structure.
\end{CommentBox}

\paragraph{The value set $R$.}
Although the definition allows for arbitrary sets $R$ of possible tensor entries,
the entries in practically relevant problems are typically real numbers ($R = \R$) or
elements of some other \emph{commutative semiring}.
In the following, we assume that $R$ is a commutative semiring $(R, \oplus, \otimes)$, where $\oplus$ is an aggregation function and $\otimes$ is a commutative combination function.
Beyond the standard arithmetic semiring $(\mathbb{R},+,*)$, other notable examples of commutative semirings are the Viterbi and Tropical semirings.

\paragraph{Scalars.}
By convention, we consider the empty product $\prod_{i=1}^{0} [d_i]$ to be the set $\{ (~) \}$ that only contains the empty tuple.
Thus, scalars fit into the above definition as tensors of order $o=0$.
Specifically, a scalar $c\in R$ is uniquely represented as the tensor $T: \{ (~) \} \to R, (~) \mapsto c$.

\begin{definition}[Delta tensors]
	For $o\in\mathbb{N}_0$, $d_1, \sdots, d_o \in\mathbb{N}$, and a semiring $R$, the tensor
	$\delta_o: [d_1] \times\sdots\times [d_o] \times [d_1] \times\sdots\times [d_o] \to R$ with entries
	\begin{equation*}
		\delta_o (p_1, \sdots, p_o, q_1, \sdots, q_o) =
		\begin{cases}
			1, \text{ if } (p_1, \sdots, p_o) = (q_1, \sdots, q_o) \\
			0, \text{ otherwise.}
		\end{cases}
	\end{equation*}
	is called a \textbf{delta tensor} of order $2o$.
	In this context, $1$ and $0$ refer to the $1$-element and $0$-element of $R$, respectively.

	Examples of delta tensors are the scalar $\delta_0 = 1$ and the unit matrix $\delta_1 = \mathds{1}$.
	The precise axis lengths $d_1, \sdots, d_o$ of a delta tensor are often omitted when they can be inferred from context.
\end{definition}

\section{Syntax}
\label{sec:syntax}

An \emph{einsum} expression has the following generic form:
$$\einsum{I_1, \sdots, I_n}{I}{T_1, \sdots, T_n},$$
where the hash symbol $\#$ replaces the word \texttt{einsum} as a more compact and recognizable marker for the beginning of an \emph{einsum} expression.
Inside the parentheses, the expression consists of three parts:
A number of index strings $I_1, \sdots, I_n$ on the left of the arrow $\rightarrow$, which we refer to as \textbf{input index strings},
one index string $I$ on the right of the arrow, which we refer to as the \textbf{output index string},
and a number of subexpressions $T_1, \sdots, T_n$ on the right of the output index string, which we refer to as \textbf{arguments}.
We also refer to the combination of input index strings, arrow, and output index string as the \textbf{format string}.

A valid \emph{einsum} expression must also satisfy the following additional constraints.
For these definitions, let $n \in \N$ be an arbitrary number of arguments, and let $\tensor[i]{T}{}$ be a tensor of order $o_i\in\N$ for all $i \in [n]$.

\paragraph{I (Index strings and sets)}
\emph{Einsum} expressions use index strings to specify how one or several input tensors are combined into a single output tensor.
Each index string consists of index symbols $s\in S$, where the potentially infinite set $S$ of all index symbols is arbitrary.
We will represent individual index symbols $s\in S$ as lowercase letters and index strings $I_i \in S^{o_i}$ as uppercase letters.
To the left of the arrow $\rightarrow$, an \emph{einsum} expression requires exactly one, possibly empty, index string $I_i = (s_{i1}, \ldots, s_{io_i})$ for each argument $\tensor[i]{T}{}$.
On its right is the output index string $I$ that corresponds to the result tensor.
The order of index symbols in an index string matters, and multiple appearances of the same index symbol are possible.
We denote the set of index symbols in an index string $I_i$ as $\sigma(I_i) \subseteq S$.

\begin{CommentBox}{The index symbol set $S$ in practice}
	Here, we assume a potentially infinite index symbol set $S$.
	In practice, the choice of an (effectively finite) index symbol set can be very important.
	For example, \emph{einsum} in PyTorch~\cite{Paszke2019} only allows $2\cdot 26 = 52$ alphabetic index symbols.
	While literal index symbols are useful for human readability, we recommend that a real-world implementation should additionally accept arbitrary integers in order to handle more axes.
\end{CommentBox}

\paragraph{II (Axis lengths)}
In an \emph{einsum} expression, each index symbol $s_{ij}$ corresponds to the $j$-th axis of the tensor $\tensor[i]{T}{}$.
For the expression to be valid, all axes corresponding to the same index symbol must match in length.
In other words, let $d_{ij} \in \N$ denote the size of the $j$-th axis of $\tensor[i]{T}{}$ for $i \in [n], j \in [o_i]$.
Then for two identical index symbols $s_{ij} = s_{i'j'}$ in different locations $i,i' \in [n], j \in [o_i], j' \in [o_{i'}]$, it must hold that
\begin{equation*}
	d_{ij} = d_{i'j'} \text{ and thus } [d_i] = [d_{i'}].
\end{equation*}
We denote the uniquely determined length of all axes that an index symbol $s = s_{ij}\in S$ corresponds to as $d_s := d_{ij}$.

\paragraph{III (Output string)}
For the result tensor of the \emph{einsum} expression to be well-defined, the lengths of its axes must be specified.
This can only be the case if every index symbol in $I$ also appears in at least one of the input index strings $I_i$.
Thus, the third condition for a valid \emph{einsum} expression is:
\begin{equation*}
	\sigma(I) \subseteq \bigcup_{i \in [n]} \sigma(I_i).
\end{equation*}

\begin{CommentBox}{Regularity of the language of \emph{einsum} format strings}
	For a finite set $S$ of index symbols, the language of \emph{einsum} format strings over $S$ is regular.
	It follows that the \emph{einsum} format strings over a finite set $S$ of index symbols can be enumerated~\cite{EnumerateContextfreeLanguages}.
\end{CommentBox}

For examples of \emph{einsum} expressions and corresponding expressions in linear algebra, see \hyperref[tab:einsum_examples]{Table~\ref*{tab:einsum_examples}}.

\begin{table*}[h!]
	\centering
	\caption[Einsum notation]{Examples of the \emph{einsum} notation and the corresponding expressions in linear algebra.}
	\label{tab:einsum_examples}
	\begin{tabular}{lcl}
		\toprule
		Operation                               & Linear algebra   & \emph{einsum}                   \\
		\midrule
		matrix product                          & $A \cdot B$      & $\einsum[\#]{ij, jk}{ik}{A, B}$ \\
		matrix transposition                    & $A^\top$         & $\einsum[\#]{ij}{ji}{A}$        \\
		elementwise product of two vectors      & $x \otimes y$    & $\einsum[\#]{i,i}{i}{x,y}$      \\
		inner product of two vectors (a scalar) & $x^\top y$       & $\einsum[\#]{i, i}{}{x, y}$     \\
		outer product of two vectors (a matrix) & $xy^\top$        & $\einsum[\#]{i, j}{ij}{x, y}$   \\
		extract the diagonal of a matrix        & \text{diag}($A$) & $\einsum[\#]{ii}{i}{A}$         \\
		broadcast a vector to a diagonal matrix & \text{diag}($v$) & $\einsum[\#]{i}{ii}{v}$         \\
		\bottomrule
	\end{tabular}
\end{table*}

\paragraph{Differences from popular libraries}

Compared to NumPy \cite{Harris2020}, PyTorch \cite{Paszke2019} and TensorFlow \cite{TensorFlow}, our definition has two significant differences:
First, the output index string always has to be specified, which is important for \emph{einsum} to be commutative, and second, the output index string can include the same index symbol more than once, which allows expressions like $\text{diag}(v) = \einsum[\#]{i}{ii}{v}$.

\section{Semantics}
\label{sec:semantics}

In this section, we define the semantics of the \emph{einsum} language.
To do so, we first introduce global positions and projections, which allow us to express from which entries in the input
tensors a given entry in the output tensor is computed.

\begin{definition}[Global positions and projections]
	Let $S = \{s_1, \ldots, s_l\}$ be the set of all index symbols used in an \emph{einsum} expression, and let $d_s$ be the length of the axis that corresponds to the index symbol $s$.
	A global position $\hat{x}$ is a tuple $(\hat{x}_{s_1}, \dots, \hat{x}_{s_l})$ with $\hat{x}_s \in [d_s]$ for each index symbol $s \in S$.
	Then $\mathcal{X} := \prod_{s \in S} [d_s]$ is the set of all global positions over $S$.
	Given an index string $I = (s'_1, \ldots, s'_o)\in S^o$, a global position $\hat{x} = (\hat{x}_{s_1}, \ldots, \hat{x}_{s_l}) \in \mathcal{X}$ can be projected to a position $(\hat{x}:I) := (\hat{x}_{s'_1}, \ldots, \hat{x}_{s'_o})$ in a particular tensor.
\end{definition}

\begin{CommentBox}{Scalars in \emph{einsum}}
	Scalar operands or results correspond to the empty index string $\lambda$ with length $0$.
	The empty index string projects every possible global position $\hat{x}$ onto the empty tuple, i.e.\ $\hat{x}:\lambda = (~)$.
\end{CommentBox}

\begin{definition}[The semantics of \emph{einsum}]
	For some $n \in \N$ and each $i \in [n]$, let $T_i$ be a tensor of order $o_i$ with a corresponding index string $I_i$.
	Further, let $I = (s_1, \dots, s_o)$ be an output index string that satisfies the conditions in \hyperref[sec:syntax]{Section~\ref*{sec:syntax}},
	let $(R, \oplus, \otimes)$ be a commutative semiring,
	and let $\mathcal{X}$ be the set of global positions.
	We define the value of the \emph{einsum} expression
	$$\tensor{T}{} = \einsum[\#]{I_1, \ldots, I_n}{I}{\tensor[1]{T}{}, \ldots, \tensor[n]{T}{}}$$
	as a sum of products in the semiring $R$.
	Specifically, $\tensor{T}{}$ is the $o$-th order tensor with the domain
	$$\domain(\tensor{T}{}) = [d_{s_1}] \times \ldots \times [d_{s_o}],$$
	where
	$$\forall x \in \domain(\tensor{T}{}): \tensor{T}{x} = \bigoplus\limits_{\substack{\hat{x} \in \mathcal{X} \\ \hat{x}:I=x}} \bigotimes\limits_{i = 1}^{n} \tensor[i]{T}{\hat{x}:I_i}.$$
\end{definition}

In words: We calculate the entry at a position $x$ by \textbf{aggregating} (e.g.\ summing over) all options $\hat{x}$ to assign values to index symbols such that the output index string $I$ projects them onto $x$.
For each such $\hat{x}$, we add the \textbf{combination} (e.g.\ product) of each respective entry in the individual input tensors.
If an \emph{einsum} expression consists entirely of scalars, then the set of used index symbols is empty.
The only possible global position in this case is $\hat{x} = (~)$, meaning that the aggregation is only over a single term (e.g.\ a product of scalars in the case of the standard arithmetic semiring).

\begin{CommentBox}{Different semirings}
	While the language itself remains agnostic of the underlying semiring (i.e.\ the aggregation and combination operations), an efficient evaluating algorithm typically does not.
	For example, when aggregating with the maximum operation, branch-and-bound algorithms might be preferable to algorithms that compute every single entry.
\end{CommentBox}

\medskip 

In machine learning, differentiability is an indispensable property because it permits efficient optimization.
But while \emph{einsum} expressions over the standard arithmetic semiring $(\R, +, \cdot)$ are differentiable, their derivative is not always another \emph{einsum} expression, but can instead require an elementwise aggregation (here: an elementwise sum) of multiple \emph{einsum} expressions.
As an addition to the \emph{einsum} language, we therefore define elementwise aggregations as follows.

\begin{definition}[Elementwise aggregation]
	Given a semiring $(R, \oplus, \otimes)$ and two tensors $S$ and $T$ over that semiring such that $\domain(S) = \domain(T)$,
	we define the \textbf{elementwise aggregate} $S\oplus T$ of these tensors as:
	$$ \forall \hat{x} \in\domain(T): (S\oplus T)(\hat{x}) = S(\hat{x}) \oplus T(\hat{x}), $$
	where the aggregation operation $\oplus$ is applied elementwise.
\end{definition}

\begin{CommentBox}{Differentiability of \emph{einsum} expressions over the arithmetic semiring}
	Expressions that use \emph{einsum} together with elementwise sums of tensors are closed under differentiation~\cite{tensorcalculus}.
	Such expressions are infinitely differentiable, but may first need to be reshaped into equivalent expressions due to technical reasons.
	Specifically, expressions like the one in our introductory example must be reshaped prior to differentiation to avoid the issue of disappearing index information.
\end{CommentBox}

\section{Algebraic properties}
\label{sec:properties}

Based on the formal definition of \emph{einsum} expressions presented in the previous sections, we now explore which syntactically different expressions are semantically equivalent.
First, we show three kinds of semantic equivalence rules as algebraic properties of \emph{einsum}: (1) \emph{Einsum} is commutative, meaning we can reorder factors, (2) \emph{einsum} is associative, meaning we can change the order of binary operations in which an expression is evaluated, and (3) \emph{einsum} is distributive with respect to the element-wise aggregation $\oplus$ of tensors.

\subsection{Commutativity}
\label{sec:properties:commutativity}

Just like in a standard product, commutativity in \emph{einsum} means that reordering factors does not change the result.
Note, however, that a `factor' in an \emph{einsum} expression consists of not only the operand tensor itself, but also its corresponding index string.

\begin{theorem}[Commutativity of \emph{einsum}]
	Given $n$ tensors $\tensor[1]{T}{}, \ldots, \tensor[n]{T}{}$ with corresponding index strings $I_1, \ldots, I_n$ and a result string $I$, we can apply any permutation $\pi$ of $n$ objects without changing the value of the \emph{einsum} expression.
	That is:
	\begin{equation*}
		\einsum{I_1, \ldots, I_n}{I}{\tensor[1]{T}{}, \ldots, \tensor[n]{T}{}} = \einsum{\pi(I_1, \ldots, I_n)}{I}{\pi(\tensor[1]{T}{}, \ldots, \tensor[n]{T}{})}
	\end{equation*}
\end{theorem}
\begin{proof}
	Let $\mathcal{X}$ be the set of global positions. Because $I$ remains unchanged, both tensors have the same domain.
	We show for every position $x$ in that domain:
	\begin{align*}
		 & \einsum{I_1, \ldots, I_n}{I}{\tensor[1]{T}{}, \ldots, \tensor[n]{T}{}}(x)                                                   \\
		 & = \bigoplus\limits_{\substack{\hat{x} \in \mathcal{X}                                                                       \\ \hat{x}:I=x}} \bigotimes\limits_{i = 1}^{n} \tensor[i]{T}{\hat{x}:I_i}
		 &                                                                                       & (\textit{einsum})                   \\
		 & = \bigoplus\limits_{\substack{\hat{x} \in \mathcal{X}                                                                       \\ \hat{x}:I=x}} \bigotimes\limits_{i = 1}^{n} \tensor[\pi(i)]{T}{\hat{x}:I_{\pi(i)}}
		 &                                                                                       & (\text{commutativity of }\otimes\,) \\
		 & = \einsum{\pi(I_1, \ldots, I_n)}{I}{\pi(\tensor[1]{T}{}, \ldots, \tensor[n]{T}{})}(x)
		 &                                                                                       & (\textit{einsum})                   \\
	\end{align*}
	In other words, \emph{einsum} is commutative because the combination operation $\otimes$ is commutative.
\end{proof}

\begin{CommentBox}{Non-commutativity in linear algebra}
	Matrix-matrix multiplication in linear algebra is not commutative.
	That is, in general, $A\cdot B\neq B\cdot A$.
	This is not in contradiction to the commutativity of \emph{einsum}, because the first matrix-matrix product is represented by the \emph{einsum} expression $\einsum{ij,jk}{ik}{A,B}$, whereas the second product is represented by $\einsum{ij, jk}{ik}{B,A}$.
	These two \emph{einsum} expressions are not generally semantically equivalent, and in fact aggregate over different axes of $A$ and $B$.
\end{CommentBox}

\subsection{Associativity}
\label{sec:properties:associativity}

For a matrix multiplication over three matrices $A,B,$ and $C$, associativity means that $(AB)\cdot C = A\cdot (BC)$.
In other words, we can first compute the matrix-matrix product $A\cdot B =: D$ and then $D\cdot C$, or we can compute $B\cdot C =: E$ first and then $A\cdot E$, without changing the result.
For an \emph{einsum} expression $$\einsum{I_1, I_2, I_3}{I}{\tensor[1]{T}{}, \tensor[2]{T}{}, \tensor[3]{T}{}},$$ the same concept applies,
meaning we can evaluate an intermediate computation over $\tensor[1]{T}{}$ and $\tensor[2]{T}{}$ first and then combine it with $\tensor[3]{T}{}$,
or we can evaluate an intermediate computation over $\tensor[2]{T}{}$ and $\tensor[3]{T}{}$ first and then combine it with $\tensor[1]{T}{}$.
More specifically:

\begin{theorem}[Associativity of \emph{einsum}]
	\label{thm:associativity}
	Given three tensors $\tensor[1]{T}{}, \tensor[2]{T}{}, \tensor[3]{T}{}$ with corresponding index strings $I_1, I_2, I_3$, and result index strings $I, I_4, I_5$ such that $\sigma(I) = \sigma(I_4) \cup \sigma(I_5)$, the following equalities hold:
	\begin{align*}
		\einsum{I_1, I_2, I_3}{I}{\tensor[1]{T}{}, \tensor[2]{T}{}, \tensor[3]{T}{}}
		 & = \einsum{I_4, I_3}{I}{\einsum{I_1, I_2}{I_4}{\tensor[1]{T}{}, \tensor[2]{T}{}}, \tensor[3]{T}{}}  \\
		 & = \einsum{I_1, I_5}{I}{\tensor[1]{T}{}, \einsum{I_2, I_3}{I_5}{\tensor[2]{T}{}, \tensor[3]{T}{}}}.
	\end{align*}
\end{theorem}
The proof follows immediately from the more fundamental nesting and denesting operations in \hyperref[sec:nesting:simple]{Section~\ref*{sec:nesting:simple}}.

\subsection{Distributivity}

Applying distributivity to elementwise aggregations of \emph{einsum} expressions can improve the computational efficiency, as can be seen in the following example from matrix algebra:
Given matrices $A,B,$ and $C \in \R^{n \times n}$, we have the two equivalent expressions
$$AB + AC =  A(B + C),$$
where the use of distributivity makes computing a second matrix product unnecessary.
Here, $+$ denotes the elementwise addition of matrices, and all products refer to the ordinary matrix-matrix multiplication.
In \emph{einsum} notation over the standard arithmetic semiring, the expressions read as
\begin{align*}
	AB + AC  & = \einsum{ik, kj}{ij}{A, B} + \einsum{ik, kj}{ij}{A, C} \text{\quad and} \\
	A(B + C) & = \einsum{ik, kj}{ij}{A, B + C}.
\end{align*}

Therefore, we next prove in general that \emph{einsum} is distributive with regard to the elementwise aggregation operation.

\begin{theorem}[Distributivity of \emph{einsum} over aggregations]
	\label{thm:distributivity}

	For each $i \in [n]$, let $T_i$ be an $o_i$-th order tensor with index string $I_i \in S^{o_i}$,
	and let $P$ and $Q$ be $o_1$-th order tensors such that $T_1 = P \oplus Q$.
	Given a result index string $I$, then $S \oplus T = U$, where
	\begin{align*}
		S & = \einsum[\#]{I_1,\ldots,I_n}{I}{P, T_2, \ldots, T_n},             \\
		T & = \einsum[\#]{I_1,\ldots,I_n}{I}{Q, T_2, \ldots, T_n}, \text{ and} \\
		U & = \einsum[\#]{I_1,\ldots,I_n}{I}{P \oplus Q, T_2, \ldots, T_n}.
	\end{align*}
\end{theorem}
\begin{proof}
	The expressions for $S, T,$ and $U$ use the same index symbols and thus have the same set $\mathcal{X}$ of global positions.
	The shared result string $I$ implies that $\domain(S \oplus T) = \domain(U)$.
	Thus, for every position $x\in\domain(U)$:
	\begin{align*}
		U(x)
		 & = \einsum[\#]{I_1,\ldots,I_n}{I}{P \oplus Q, T_2, \ldots, T_n}                                                    \\
		 & = \bigoplus\limits_{\substack{\hat{x} \in \mathcal{X}                                                             \\
				\hat{x}:I=x}} (P+Q)(\hat{x}:I_1) \otimes \bigotimes\limits_{i = 2}^{n} T_i(\hat{x}:I_i)
		 &                                                                & (\text{\emph{einsum}})                           \\
		 & = \bigoplus\limits_{\substack{\hat{x} \in \mathcal{X}                                                             \\
				\hat{x}:I=x}} \left( P(\hat{x}:I_1) \otimes \bigotimes\limits_{i = 2}^{n} T_i(\hat{x}:I_i)
		\oplus Q(\hat{x}:I_1) \otimes \bigotimes\limits_{i = 2}^{n} T_i(\hat{x}:I_i) \right)
		 &                                                                & (\text{distrib.\ of }\oplus\text{ and } \otimes) \\
		 & = \bigoplus\limits_{\substack{\hat{x} \in \mathcal{X}                                                             \\ \hat{x}:I=x}}
		\left(P(\hat{x}:I_1) \otimes \bigotimes\limits_{i = 2}^{n} T_i(\hat{x}:I_i)\right)
		\oplus
		\bigoplus\limits_{\substack{\hat{x} \in \mathcal{X}                                                                  \\ \hat{x}:I=x}}
		\left(Q(\hat{x}:I_1) \otimes \bigotimes\limits_{i = 2}^{n} T_i(\hat{x}:I_i)\right)
		 &                                                                & (\text{commutativity of }\oplus)                 \\
		 & = S(x) \oplus T(x)
		 &                                                                & (\text{\emph{einsum}})
	\end{align*}
\end{proof}

\section{Nesting and denesting}
\label{sec:nesting}

A key step when optimizing \emph{einsum} expressions for efficient evaluation is the computation of a good \emph{contraction path}.
A contraction path decomposes a single \emph{einsum} expression over $n$ tensors into $n - 1$ binary \emph{einsum} expressions, meaning they only have two input tensors each.
For example, the matrix-matrix-vector product $A\cdot B\cdot v$ from the introduction can be written as the following \emph{einsum} expression:

$$\einsum{ij, jk, k}{i}{A, B, v}$$

Associativity allows us to evaluate this expression in two different ways.
One possibility is to compute the matrix-matrix product $A\cdot B$ first, which corresponds to the nested \emph{einsum} expression

$$\einsum{ik, k}{i}{\einsum{ij, jk}{ik}{A, B}, v}.$$

The other possibility is to compute the matrix-vector product $B\cdot v$ first, which results in a different nested \emph{einsum} expression,

$$\einsum{ij, j}{i}{A, \einsum{jk, k}{j}{B, v}}.$$

Of these two possibilities, the second is computationally more efficient, because it decomposes the operation into two matrix-vector products, completely avoiding the matrix-matrix product.

In this example, we used the associativity of matrix products in linear algebra.
Below, we will prove that all \emph{einsum} expressions can similarly be decomposed into a series of multiple smaller expressions.

\begin{CommentBox}{Hardness of optimizing contraction paths}
	Generally, the choice of contraction path has a significant impact on the computational complexity of the expression.
	Optimizing the contraction path is an NP-hard problem \cite{Xu2023} and beyond the scope of this work.
\end{CommentBox}

\subsection{Restricted nesting and denesting}
\label{sec:nesting:simple}

In the following, we establish the soundness of the transformations used in contraction paths, which we call \textbf{restricted nesting and denesting}.
This rule is restricted in the sense that there are nested \emph{einsum} expressions that cannot be denested into one flat expression using only this rule.
Nevertheless, the rule covers all transformations required to specify contraction paths, i.e.\ to turn a single $n$-ary \emph{einsum} expression into a hierarchy of nested binary \emph{einsum} expressions.

\begin{theorem}[Restricted nesting and denesting of \emph{einsum} expressions]
	\label{thm:nesting:simple}
	For each $i \in [m + n]$, let $T_i$ be an $o_i$-th order tensor with an index string $I_i$ of matching length, let $I_u, I_v$ also be index strings, and let
	$$U = \einsum[\#]{I_1,\ldots,I_m}{I_u}{T_1,\ldots,T_m}$$
	and
	$$V = \einsum[\#]{I_u, I_{m+1}, \ldots , I_{m+n}}{I_v}{U, T_{m+1}, \ldots, T_{m+n}}.$$
	If the inner and the outer \emph{einsum} expression do not have any index symbols in common except for those that appear in $I_u$, then
	$V = \einsum[\#]{I_1, \ldots, I_{m + n}}{I_v}{T_1, \ldots, T_{m + n}}$.
\end{theorem}
\begin{proof}
	Let $\mathcal{X}_U, \mathcal{X}_V$ be the global position sets over the index symbol sets $S_U,S_V$ for the expressions $U$ and $V$, respectively.
	Let $S := S_U \cup S_V$ be the set of index symbols occurring in either expression, and let $\mathcal{X} := \prod_{s \in S} [d_s]$ be the global position set over these index symbols.
	That $U$ and $V$ have no other symbols in common than those in $I_u$ can be formally stated as:
	$$\Big(\bigcup_{i=1}^m \sigma(I_i)\Big) \cap \Big( \bigcup_{j=m+1}^{m+n} \sigma(I_j)\Big) \subseteq \sigma(I_u).$$
	Then for all positions $x\in\domain(V)$:
	\begin{align*}
		V(x)
		                                                                                                                & = \einsum[\#]{I_u, I_{m+1}, \ldots , I_{m+n}}{I_v}{U, T_{m+1}, \ldots, T_{m+n}}(x)                                                           \\
		                                                                                                                & = \bigoplus\limits_{\substack{\hat{x}_v \in \mathcal{X}_V                                                                                    \\ \hat{x}_v:I_v=x}}
		U(\hat{x}_v:I_u) \otimes \bigotimes\limits_{i=m+1}^{m+n} T_i(\hat{x}_v:I_i)                                     &                                                                                    & (\text{\emph{einsum} for $V$})                          \\
		                                                                                                                & = \bigoplus\limits_{\substack{\hat{x}_v \in \mathcal{X}_V                                                                                    \\ \hat{x}_v:I_v=x}}
		\bigg(\bigoplus\limits_{\substack{\hat{x}_u \in \mathcal{X}_U                                                                                                                                                                                                  \\ \hat{x}_u:I_u=\hat{x}_v:I_u}}
		\bigotimes\limits_{j=1}^{m} T_j(\hat{x}_u:I_j)\bigg) \otimes \bigotimes\limits_{i=m+1}^{m+n} T_i(\hat{x}_v:I_i) &                                                                                    & (\text{\emph{einsum} for $U$})                          \\
		                                                                                                                & = \bigoplus\limits_{\substack{\hat{x}_v \in \mathcal{X}_V                                                                                    \\ \hat{x}_v:I_v=x}}
		\bigoplus\limits_{\substack{\hat{x}_u \in \mathcal{X}_U                                                                                                                                                                                                        \\ \hat{x}_u:I_u=\hat{x}_v:I_u}} \bigg( \bigotimes\limits_{j=1}^{m} T_j(\hat{x}_u:I_j) \otimes \bigotimes\limits_{i=m+1}^{m+n} T_i(\hat{x}_v:I_i)\bigg) && (\text{semiring distributivity})\\
		                                                                                                                & = \bigoplus\limits_{\substack{\hat{x} \in \mathcal{X}                                                                                        \\ \hat{x}:I_v=x}} \bigotimes\limits_{i=1}^{m+n} T_i(\hat{x}:I_i) && (\ast\text{ see below})\\
		                                                                                                                & = \einsum[\#]{I_1, \ldots, I_{m + n}}{I_v}{T_{1}, \ldots, T_{m + n}}(x)            &                                & (\text{\emph{einsum}}) \\
	\end{align*}
	$(\ast)$ Observe that for every pair of global positions $\hat{x}_u \in \mathcal{X}_U$ and $\hat{x}_v \in \mathcal{X}_V$ with $\hat{x}_u:I_u = \hat{x}_v:I_u$,
	there exists a unique $\hat{x} \in \mathcal{X}$ that assigns the same value to each index symbol in $S=S_U \cup S_V$ as $\hat{x}_u$ and $\hat{x}_v$, because the shared index symbols $S_U \cap S_V = \sigma(I_u)$ are each assigned the same value by both $\hat{x}_u$ and $\hat{x}_v$.
	The combinations of all possible pairs $\hat{x}_u, \hat{x}_v$ exactly make up all possible global positions $\hat{x} \in \mathcal{X}$.
	Thus, both sums are identical.
\end{proof}

\begin{CommentBox}{Associativity}
	Since the associativity rules satisfy the condition of \hyperref[sec:nesting:simple]{Section~\ref*{sec:nesting:simple}},
	it follows directly by denesting
	$$
		\einsum{I_4, I_3}{I}{\einsum{I_1, I_2}{I_4}{\tensor[1]{T}{}, \tensor[2]{T}{}}, \tensor[3]{T}{}} \text{\quad and\quad}
		\einsum{I_1, I_5}{I}{\tensor[1]{T}{}, \einsum{I_2, I_3}{I_5}{\tensor[2]{T}{}, \tensor[3]{T}{}}}
	$$
	into
	$\einsum{I_1, I_2, I_3}{I}{\tensor[1]{T}{}, \tensor[2]{T}{}, \tensor[3]{T}{}}$.
\end{CommentBox}

\subsection{General nesting and denesting}
\label{sec:nesting:general}

Restricted nesting and denesting as defined in \hyperref[thm:nesting:simple]{Theorem~\ref*{thm:nesting:simple}} is always sufficient when starting with a completely flat expression.
In practice, however, one may encounter nested expressions that do not conform to the restrictions in \hyperref[thm:nesting:simple]{Theorem~\ref*{thm:nesting:simple}} regarding the index symbols in the inner and outer expression.
For example, a matrix-matrix multiplication $A \cdot \text{diag}(v)$ of a matrix $A \in \R^n$ and a diagonal matrix $\text{diag}(v) \in \R^{n \times n}$ can be expressed by a nested \emph{einsum} expression
$$\einsum{ik,kj}{ij}{A, \text{diag}(v)} = \einsum{ik,kj}{ij}{A, \einsum{i}{ii}{v}}.$$
This expression cannot be denested with the restricted rule, because the output string $ii$ of the inner expression does not match the input string $ik$ of the outer expression, which is required by  \hyperref[thm:nesting:simple]{Theorem~\ref*{thm:nesting:simple}}. However, a semantically equivalent denested expression
$$\einsum{ik, k}{ik}{A, v}$$
does exist and is even computationally less expensive, because it avoids the creation of the diagonal matrix and the evaluation of a matrix-matrix product.

In general, the two restrictions of \hyperref[thm:nesting:simple]{Theorem~\ref*{thm:nesting:simple}} can lead to two obstructions when denesting.
The first issue is index symbol collisions, which occur where the outer and the inner expression use \emph{the same} index symbol for \emph{different} purposes.
For example, in the expression
$$\einsum{ijk, }{}{U, \einsum{ijk}{}{V}},$$
the same index symbols $i,j$ and $k$ are used for both of the tensors $U$ and $V$.
To retain the same semantics, the denested expression requires a separate set of symbols for each tensor:
$$\einsum{ijk, abc}{}{U, V}$$
This problem occurs whenever the outer and the inner expression have an index symbol in common that does not appear in the shared index string (called $I_u$ in \hyperref[thm:nesting:simple]{Theorem~\ref*{thm:nesting:simple}}),
but it can always be resolved by renaming index symbols in either expression until all collisions are eliminated.
Because the scope of a given index symbol is only a single \emph{einsum} expression, renaming every occurrence within that expression with a previously unused symbol leaves the result unchanged.

The second issue is that \hyperref[thm:nesting:simple]{Theorem~\ref*{thm:nesting:simple}} assumes that the result index string of the inner expression and the corresponding operand index string of the outer expression match exactly, which may not be the case.
If one or both index strings include duplicate symbols, this cannot be resolved by simply renaming index symbols, because renaming only works when introducing index symbols that are \emph{not already in use} within the expression.
This is, for instance, the case in the example
$$\einsum{ik,kj}{ij}{A, \text{diag}(v)} = \einsum{ik,kj}{ij}{A, \einsum{i}{ii}{v}}$$
from above.
To resolve this second issue, we replace a direct renaming scheme with the more general concept of a symbol map.

\begin{definition}[Index symbol map]
	An index symbol map for an index symbol set $S$ is a mapping $\nu:S \to \hat S$, where $\hat S$ is another index symbol set with $|\hat S| \leq |S|$.
	The extension of $\nu$ to entire index strings is denoted as $\nu^\ast$.
	That is, an index string $I=(s_1, \ldots, s_o)\in S^o$ of length $o$ is mapped to $\nu^\ast (I) = (\nu(s_1), \ldots, \nu(s_o))$.
\end{definition}

Not every symbol map can be used to denest a nested \emph{einsum} expression.
We derive constraints on the symbol map from the inner and outer expression and encode them in an \emph{index symbol graph}.

\begin{definition}[Index symbol graph]
	For a given nested \emph{einsum} expression
	$$\einsum{\hat{I}_u, I_{m+1}, \ldots , I_{m+n}}{I_v}{\einsum{I_1,\ldots,I_m}{I_u}{T_1,\ldots,T_m}, T_{m+1}, \ldots, T_{m+n}},$$
	where the $T_i, i \in [m + n]$ are tensors with index strings $I_i$, and
	$I_u = (a_1, \ldots, a_d), \hat{I}_u = (b_1, \ldots, b_d)$ and $I_v$ are index strings,
	the vertex set $V$ of the undirected index symbol graph is given by
	$$\{ u_1, \ldots, u_d\} \cup \{ v_1, \ldots, v_d\} \cup \{ x_1, \ldots, x_d\},$$
	and the edge set $E$ is derived from the nested expression as follows:
	\begin{align*}
		 & (u_i, u_j)\in E\: \text{ for all } i,j\in[d] \text{ such that }\: a_i = a_j,             \\
		 & (v_i, v_j)\in E\: \text{ for all } i,j\in[d] \text{ such that }\: b_i = b_j, \text{ and} \\
		 & (u_i, x_i),\, (v_i, x_i) \in E \text{ for all } i\in[d].
	\end{align*}
\end{definition}

The vertices in an index symbol graph represent index symbols, and edges connect vertices that must be mapped to the same index symbol.
We get a valid symbol map by mapping every symbol from the original index symbol set $S$ that is assigned to some vertex in a connected component of the graph to the new symbol in $\hat S$ for the component.
Note that if the nested expression has index duplications in neither the inner nor the outer expression, the graph has exactly one connected component for every label, and thus the symbol map amounts to a simple renaming of all index symbols.

\begin{CommentBox}{Denesting in practice}
	In practice, the index symbol graph can be built without the vertices $x_i, i\in [d]$.
	The symbol map could then explicitly be stated by identifying the graph components.
	We have included these vertices because we needed them in the correctness proof of \hyperref[thm:nesting:general]{Theorem~\ref*{thm:nesting:general}} that is stated below.
\end{CommentBox}

\begin{theorem}[General nesting and denesting]
	\label{thm:nesting:general}
	Let $V$ be a nested \emph{einsum} expression
	$$V = \einsum{\hat{I}_u, I_{m+1}, \ldots , I_{m+n}}{I_v}{\einsum{I_1,\ldots,I_m}{I_u}{T_1,\ldots,T_m}, T_{m+1}, \ldots, T_{m+n}},$$
	where the $T_i, i \in [m + n]$ are tensors with index strings $I_i$, and
	$I_u = (a_1, \ldots, a_d), \hat{I}_u = (b_1, \ldots, b_d)$ and $I_v$ are index strings.
	If $\nu$ is an index symbol map that respects the restrictions encoded in the index symbol graph constructed from $V$, then
	$$V = \einsum{\nu^\ast(I_1), \ldots, \nu^\ast(I_{m + n})}{\nu^\ast(I_v)}{T_1, \ldots, T_{m + n}}. $$
\end{theorem}

The theorem suggests a straightforward denesting algorithm in four steps:
(1) construct the index symbol graph,
(2) determine its connected components,
(3) derive the symbol map $\nu$ from the connected components, and
(4) construct the denested expression by using $\nu^\ast$ to rename index symbols.
The idea is illustrated on the following example that cannot be denested using only \hyperref[thm:nesting:simple]{Theorem~\ref*{thm:nesting:simple}}:
$$\vphantom{\overbrace{abbcde}^{\hat I_u}}\einsum[\#]{a,b,c,d,e,\smash{\overbrace{abbcde}^{\hat I_u}}}{bc}{v_1, v_2, v_3, v_4, v_5, \einsum[\#]{i,j,k,l}{\smash{\overbrace{iijkkl}^{I_u}}}{v_6, v_7, v_8, v_9}},$$
where the vectors $v_i \in \R^{d_{i}}, i \in \{1,\ldots,9\}$ have axis lengths $d_i$ that match as required by \emph{einsum}.
\begin{figure}[h]
	\centering
    \includegraphics{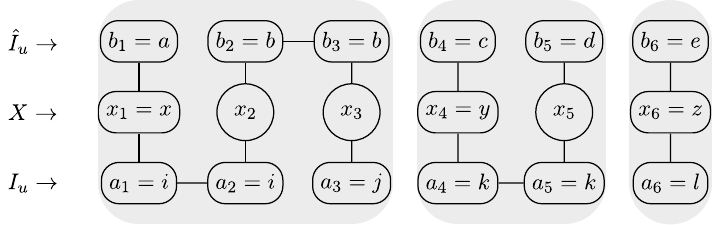}
	\caption{The index symbol graph for our example that cannot be denested directly by \hyperref[thm:nesting:simple]{Theorem~\ref*{thm:nesting:simple}}.}
	\label{fig:nested_expressions:example_index_graph}
\end{figure}
The index symbol graph for the example is shown in \hyperref[fig:nested_expressions:example_index_graph]{Figure~\ref*{fig:nested_expressions:example_index_graph}}.
From the connected components of the graph, we derive the new index symbols $x_1=x, x_4=y$, and $x_6=z$ for the denested expression, and the index symbol map $\nu$ that maps $a,b,i$, and $j$ to $x$, maps $c,d$, $k$ to $y$, and maps $e$ and $l$ to $z$.
Using this symbol map, we get the denested expression
$$\einsum[\#]{x,x,y,y,z,x,x,y,z}{xy}{v_1, v_2, v_3, v_4, v_5, v_6, v_7, v_8, v_9}.$$

To formally prove the correctness of the denesting algorithm and thus \hyperref[thm:nesting:general]{Theorem~\ref*{thm:nesting:general}}, we show that the same result expression can be obtained by sequentially applying three rewriting steps, each of which preserves semantic equivalence.
The three rewriting steps are:
(1) introduce delta tensors as additional operands into an \emph{einsum} expression to replace any one index symbol with a new index symbol,
(2) apply restricted denesting using \hyperref[thm:nesting:simple]{Theorem~\ref*{thm:nesting:simple}},
(3) merge the previously introduced index symbols in accordance with the connected components of the index symbol graph.

Next, we formally define the splitting and merging operations, prove that these operations keep the transformed expressions semantically equivalent, and finally prove that the three rewriting steps reach the same result as our renaming algorithm.

\begin{definition}[Delta split and delta merge]
	Given an \emph{einsum} expression
	$$\einsum{I_1, \ldots, I_n}{I}{\tensor[1]{T}{}, \ldots, \tensor[n]{T}{}}$$
	with corresponding index strings $I_i$ and operand tensors $T_i$ for all $i\in[n]$, let $a$ be an index symbol that appears in at least one of the index strings $I_i$ and let $b$ be a new index symbol that does not appear in any of the index strings.
	Further, let
	$$\einsum{ab, \hat{I}_1, \ldots, \hat{I}_n}{\hat{I}}{\delta_1, \tensor[1]{T}{}, \ldots, \tensor[n]{T}{}}$$
	be a second \emph{einsum} expression in which the index strings $\hat{I}_i$ and the index string $\hat{I}$ each result from the corresponding index strings $I_i$ and $I$
	by replacing all, some, or none of the occurrences of $a$ with $b$.
	We call the reshaping operation that introduces a delta tensor to split the index symbol $a$ into two separate symbols $a$ and $b$ a \textbf{delta split} (or $\delta$-split).
	The inverse operation, which removes the delta tensor and merges all occurrences of the index symbol $b$ by replacing them with $a$, is called a \textbf{delta merge} (or $\delta$-merge).
\end{definition}

\begin{CommentBox}{Delta switch}
	Because the unit matrix $\delta_1 = \mathds{1}$ is symmetric, and thus $\delta_1(ab)=\delta_1(ba)$ for all $a$ and $b$,
	a $\delta$-split can equivalently introduce, and a $\delta$-merge remove, the left index symbol of the delta tensor instead of the right.
\end{CommentBox}

\begin{theorem}
	Applying a $\delta$-split or $\delta$-merge
	results in a semantically equivalent \emph{einsum} expression.
\end{theorem}
\begin{proof}
	Let
	$$E = \einsum{I_1, \ldots, I_n}{I}{\tensor[1]{T}{}, \ldots, \tensor[n]{T}{}} \text{\quad and\quad} F = \einsum{ab, \hat{I}_1, \ldots, \hat{I}_n}{\hat{I}}{\delta_1, \tensor[1]{T}{}, \ldots, \tensor[n]{T}{}}.$$
	We have $d_a=d_b$ for the lengths of the axes of $\delta_1$ corresponding to $a$ and $b$, because $\delta_1$ is the unit matrix and thus symmetric.
	Therefore, $\domain(E) = \domain(F)$ still holds even if $\hat{I} \neq I$.

	Let $\mathcal{X}_E,\mathcal{X}_F$ be the sets of global positions over the respective index symbols $S_E,S_F$ for the expressions $E$ and $F$.
	We show for every position $x\in\domain(F)$:
	\begin{align*}
		F(x)
		 & = \einsum{ab, \hat{I}_1, \ldots, \hat{I}_n}{\hat{I}}{\delta_1, \tensor[1]{T}{}, \ldots, \tensor[n]{T}{}}                      \\
		 & = \bigoplus\limits_{\substack{\hat{x}_f \in \mathcal{X}_F                                                                     \\ \hat{x}_f:\hat{I}=x}} \delta_1(\hat{x}_f:ab) \bigotimes\limits_{i = 1}^{n} T_i(\hat{x}_f:\hat{I}_i) && (\emph{einsum}) \\
		 & = \bigoplus\limits_{\substack{\hat{x}_f \in \mathcal{X}_F                                                                     \\ \hat{x}_f:\hat{I}=x \\ \hat{x}_f:a = \hat{x}_f:b}} \delta_1(\hat{x}_f:ab) \bigotimes\limits_{i = 1}^{n} T_i(\hat{x}_f:\hat{I}_i) \\
		 & \quad \oplus \bigoplus\limits_{\substack{\hat{x}_f \in \mathcal{X}_F                                                          \\ \hat{x}_f:\hat{I}=x \\ \hat{x}_f:a \neq \hat{x}_f:b}} \delta_1(\hat{x}_f:ab) \bigotimes\limits_{i = 1}^{n} T_i(\hat{x}_f:\hat{I}_i) && (\emph{commutativity of }\oplus) \\
		 & = \bigoplus\limits_{\substack{\hat{x}_f \in \mathcal{X}_F                                                                     \\ \hat{x}_f:\hat{I}=x \\ \hat{x}_f:a = \hat{x}_f:b}} 1 \otimes \bigotimes\limits_{i = 1}^{n} T_i(\hat{x}_f:\hat{I}_i)
		\oplus \bigoplus\limits_{\substack{\hat{x}_f \in \mathcal{X}_F                                                                   \\ \hat{x}_f:\hat{I}=x \\ \hat{x}_f:a \neq \hat{x}_f:b}} 0 \otimes \bigotimes\limits_{i = 1}^{n} T_i(\hat{x}_f:\hat{I}_i) && (\emph{definition of }\delta_1) \\
		 & = \bigoplus\limits_{\substack{\hat{x}_f \in \mathcal{X}_F                                                                     \\ \hat{x}_f:\hat{I}=x \\ \hat{x}_f:a = \hat{x}_f:b}} \bigotimes\limits_{i = 1}^{n} T_i(\hat{x}_f:\hat{I}_i)  && (\text{neutral element}) \\
		 & = \bigoplus\limits_{\substack{\hat{x}_e \in \mathcal{X}_E                                                                     \\ \hat{x}_e:I=x}} \bigotimes\limits_{i = 1}^{n} T_i(\hat{x}_e:I_i) && (\ast\text{ see below})                              \\
		 & = \einsum{I_1, \ldots, I_n}{I}{\tensor[1]{T}{}, \ldots, \tensor[n]{T}{}}(x)                              &  & (\emph{einsum}) \\
		 & = E(x)
	\end{align*}
	$(\ast)$ This equality switches between the index symbol sets $S_E$ and $S_F = S_E \cup \{ b \}$.
	Given the restrictions $\hat{x}_f:\hat{I} = x = \hat{x}_e:I$ and $\hat{x}_f:a = \hat{x}_f:b$,
	each $\hat{x}_e \in \mathcal{X}_E$ (a global position in $E$) corresponds \emph{uniquely} to an $\hat{x}_f \in \mathcal{X}_F$ (a global position in $F$) and vice-versa.
	Thus, the sums are equal.
\end{proof}

Now, we are prepared to prove the correctness of our renaming algorithm, and thus \hyperref[thm:nesting:general]{Theorem~\ref*{thm:nesting:general}} by showing that our application of the formally verified rewriting rules reaches the same result.
\begin{proof}
	Recall the three rewriting steps: $\delta$-splitting, restricted denesting according to \hyperref[thm:nesting:simple]{Theorem~\ref*{thm:nesting:simple}}, and $\delta$-merging.
	We start with the nested expression
	$$V = \einsum{\hat{I}_u, I_{m+1}, \ldots , I_{m+n}}{I_v}{U, T_{m+1}, \ldots, T_{m+n}},$$
	where
	$$U = \einsum{I_1,\ldots,I_m}{I_u}{T_1,\ldots,T_m}.$$
	We can assume that the inner and outer expressions have no index symbols in common. Otherwise, we simply replace index symbols in either expression until this condition is satisfied.

	(1) \emph{$\delta$-splitting}: Using $I_u = (a_1, \ldots, a_d)$, $\hat{I}_u = (b_1, \ldots, b_d)$, and $X=(x_1, \ldots, x_d)$, which is an index string with new, pairwise distinct index symbols $x_i$, we apply $d$ $\delta$-splits in each expression to replace both $I_u$ and $\hat{I}_u$ with $X$.
	This yields the intermediary expressions
	$$U = \einsum{a_1 x_1, \ldots, a_d x_d, I_1, \ldots, I_m}{X}{\delta_1, \ldots, \delta_1, T_1,\ldots,T_m}$$
	and
	$$V = \einsum{b_1 x_1, \ldots, b_d x_d, X, I_{m+1}, \ldots , I_{m+n}}{I_v}{\delta_1, \ldots, \delta_1, U, T_{m+1}, \ldots, T_{m+n}}.$$
	Here, we only rename exactly one occurrence of the respective index symbol $a_i$ or $b_i$ with each $\delta$-split, meaning that all index strings other than $I_u$ and $\hat{I}_u$ remain unchanged.

	(2) \emph{Restricted denesting}: Now, the result index string in the inner expression and the corresponding input index string in the outer expression match.
	Using \hyperref[thm:nesting:simple]{Theorem~\ref*{thm:nesting:simple}}, we get:
	$$V = \einsum{b_1 x_1, \ldots, b_d x_d, a_1 x_1, \ldots, a_d x_d, I_1,\ldots, I_{m+n}}{I_v}{\delta_1, \ldots, \delta_1, T_1,\ldots,\ldots, T_{m+n}}$$

	(3) \emph{$\delta$-merging}: The expression has already been denested after the previous step. However, we still want to simplify the expression by getting rid of the delta tensors. Therefore, we iteratively merge the $2d$ delta tensors which have been introduced by $\delta$-splitting.
	For each index string $a_i x_i$ or $b_i x_i$ corresponding to a delta tensor, we retain the index symbol $x_i$ when merging.
	Index strings $x_i x_j$ can occur as intermediary results if the $a_i$ or $b_i$ contain duplicate index symbols.
	In these cases, we retain the index symbol $x_{\min(i,j)}$.

	For the proof of \hyperref[thm:nesting:general]{Theorem~\ref*{thm:nesting:general}}, it remains to be shown that the expression after the $\delta$-merging step has the form
	$$V = \einsum{\nu^\ast(I_1), \ldots, \nu^\ast(I_{m + n})}{\nu^\ast(I_v)}{T_1, \ldots, T_{m + n}},$$
	where $\nu$ is an index symbol map that satisfies the conditions encoded in the index symbol graph of the original nested expression, and $\nu^\ast$ is the extension of $\nu$ to index strings.

	This is the case because the graph edges that connect symbols in the same position $i\in[d]$ are represented by the $\delta$-tensors with index strings $a_i x_i$ and $b_i x_i$,
	while the graph edges between nodes that contain the same index symbol $s$ at different positions $i,j\in[d]$ are represented by that symbol appearing in the index strings of multiple $\delta$-tensors,
	namely $s x_i$ and $s x_j$.
	Thus, the merging step implicitly computes the connected components of the index symbol graph.
\end{proof}

\section{Auxiliary simplification rules}
\label{sec:auxiliary}

In this section, we provide a few additional equivalence rules for \emph{einsum} expressions.

A semiring $(R, \oplus, \otimes)$ has a $1$-element that is the neutral element of its multiplication $\otimes$.
Identifying a similar neutral element in an \emph{einsum} expression is more difficult.
For example, both delta tensors and constant all-ones tensors, i.e., tensors in which the entry at every position is $1$, can potentially be `neutral' operands, depending on their index strings.
Reshaping an \emph{einsum} expression by introducing or removing such neutral operands matters,
for example, when computing the symbolic derivative of an \emph{einsum} expression \cite{tensorcalculus}.
During automatic symbolic differentiation, unnecessary delta tensors are systematically created, and all-ones tensors may have to be introduced in order to avoid the issue of disappearing index information, as mentioned in the introduction.

The following proofs demonstrate when and how \emph{einsum} expressions involving delta tensors and constant tensors (including all-ones tensors) can be equivalently reshaped.
In particular, we conclude that there is always a way to omit delta tensors, which is not true in other versions of \emph{einsum} such as the one used by NumPy.
Another simple yet important semantic equivalence is the identity operation.
The following semantic equivalence rule can be used as a final simplification step after applying other rules and arriving at the most basic \emph{einsum} expression $\einsum{I}{I}{T}$.
\begin{lemma}[Identity operation]
	\label{lem:identityOp}
	Let $T$ be a tensor of order $o\in\N$.
	If $I$ is an index string of length $o$ with \emph{pairwise distinct} index symbols, then $T = \einsum{I}{I}{T}$.
\end{lemma}
\begin{proof}
	Let $\mathcal{X}$ be the set of global positions over the index symbols in $I$. Evidently, $\domain \left(\einsum{I}{I}{T}\right) = \domain (T) = \mathcal{X}$.
	At every position $x \in\domain (T)$:
	$$
		\einsum{I}{I}{T} (x)
		= \bigoplus\limits_{\substack{\hat{x}\in\mathcal{X} \\ \hat{x}:I=x}} T(\hat{x}:I)
		= T(x)
	$$
\end{proof}
Note that if we have duplicate index symbols in the index string $I$, then $\mathcal{X} \neq \domain(T)$.
For example, for $I = (i, i)$ and a corresponding axis length $d_i$, we have $\mathcal{X} = [d_i]$ but $\domain(T) = [d_i]\times[d_i]$.
As a consequence, the aggregation can be empty and thus for $\einsum{I}{I}{T}$ to have additional $0$ entries.

\begin{lemma}[Neutral combination]
	\label{lem:neutral_combination}
	Let $T$ be a tensor of order $o\in\N$ and $I, J$ be index strings such that $\sigma(J)\subseteq\sigma(I)$.
	Then:
	$$\einsum{I, J}{I}{T, \mathbf{1}_J} = \einsum{I}{I}{T},$$
	where $\mathbf{1}_J$ is a tensor that takes the scalar value $1$ at every position.
\end{lemma}
\begin{proof}
	Let $U = \einsum{I,J}{I}{T, \mathbf{1}_J}$ and $V = \einsum{I}{I}{T}$.
	Further, let $\mathcal{X}$ be the set of global positions over the index symbols in $I$.
	Evidently, $\domain(U) = \domain(V)$.
	We show for every position $x\in\domain(U)$:
	\begin{equation*}
		U(x)
		= \bigoplus\limits_{\substack{\hat{x}\in\mathcal{X} \\ \hat{x}:I=x}} T(\hat{x}:I) \otimes \mathbf{1}_J(\hat{x}:J)
		= \bigoplus\limits_{\substack{\hat{x}\in\mathcal{X} \\ \hat{x}:I=x}} T(\hat{x}:I) \otimes 1
		= V(x)
	\end{equation*}
\end{proof}
Note: If the all-ones tensor introduces new index symbols to the expression (i.e.\ $\sigma(J) \not\subseteq \sigma(I)$), the same does not hold, because that would add additional global positions to $\mathcal{X}$.

\begin{lemma}[Constant vectorization]
	\label{lem:vectorization}
	Let $C$ be a tensor of order $o$ with the same entry $c$ in every position $(x_1, \sdots, x_o)\in\domain(C)$.
	Then for an index string $I=(i_1, \dots, i_o)$, it holds that
	$$\einsum{~, i_1, \sdots, i_o}{I}{c, \mathbf{1}_1, \sdots, \mathbf{1}_o} = \einsum{I}{I}{C},$$
	where the vectors $\mathbf{1}_i$ are all-ones vectors with lengths corresponding to the axes of $C$.
\end{lemma}
\begin{proof}
	Let $U = \einsum{, i_1, \sdots, i_o}{I}{c, \mathbf{1}_1, \sdots, \mathbf{1}_o}$ and $V = \einsum{I}{I}{C}$.
	Further, let $\mathcal{X}$ be the set of global positions over the index symbols in $I$.
	Evidently, $\domain(U) = \domain(V)$.
	At every position $x = (x_1, \sdots, x_o) \in\domain(U)$:
	$$
		U(x) =
		\bigoplus\limits_{\substack{\hat{x}\in\mathcal{X} \\ \hat{x}:I=x}} c(\hat{x}:()) \otimes \bigotimes\limits_{j=1}^o \mathbf{1}_j(\hat{x}:i_j)
		= \bigoplus\limits_{\substack{\hat{x}\in\mathcal{X} \\ \hat{x}:I=x}} c \otimes 1
		= \bigoplus\limits_{\substack{\hat{x}\in\mathcal{X} \\ \hat{x}:I=x}} C(\hat{x}:I)
		= V(x).
	$$
\end{proof}

\begin{lemma}[Delta substitution]
	\label{lem:delta_substitution}
	Let $o\in\mathbb{N}$ and let $I$ be an index string of length $o$ with pairwise distinct index symbols.
	Further, let $\delta_o$ be a delta tensor of order $o$. Then
	$\delta_o = \einsum[\#]{I}{II}{\mathbf{1}_o}$,
	where $\mathbf{1}_o$ is an all-ones tensor of order $o$ with axis lengths matching $\delta_o$.
\end{lemma}
\begin{proof}
	Let $\mathcal{X}$ be the set of global positions over the index symbols in $I$.
	For all positions $x = (x_1, \sdots, x_{2o}) \in \domain(\delta_o) = \domain\left(\einsum[\#]{I}{II}{\mathbf{1}_o}\right)$, it holds that
	\begin{align*}
		\einsum[\#]{I}{II}{\mathbf{1}_o} (x)
		 & = \sum\limits_{\substack{\hat{x}\in\mathcal{X}                                            \\ \hat{x}:II=x}} \mathbf{1}_o(\hat{x}:I)  && (\text{\emph{einsum}}) \\
		 & =  \begin{cases}
			      1, (x_1, \sdots, x_o) = (x_{o+1}, \sdots, x_{2o}) \\
			      0, \text{otherwise}
		      \end{cases} &  &
		\begin{array}{@{}l}
			(\text{there is exactly one such }\hat{x}.) \\
			(\text{there is no such }\hat{x}.)
		\end{array}                                                \\
		 & = \delta_o (x)                                       &  & (\text{definition of }\delta_o)
	\end{align*}
\end{proof}

\begin{corollary}[Delta (de)composition]
	For $o\geq 1$, it is possible to construct any delta tensor $\delta_o$ from $o$ unit matrices $\delta_1$.
\end{corollary}
\begin{proof}
	Follows directly by combining \hyperref[lem:vectorization]{Lemma~\ref*{lem:vectorization}} and \hyperref[lem:delta_substitution]{Lemma~\ref*{lem:delta_substitution}}.
\end{proof}

\begin{corollary}[Delta removal]
	Any delta tensor or constant tensor can be removed from an \emph{einsum} expression, leaving only a single scalar (which can be omitted if it is $1$) and exactly one all-ones vector for every index symbol that does not appear elsewhere in the \emph{einsum} expression.
\end{corollary}
\begin{proof}
	Any constant tensor can be reduced to a scalar and ones-vectors by applying \hyperref[lem:vectorization]{Lemma~\ref*{lem:vectorization}} and denesting.
	All of the all-ones vectors that do not introduce a new axis can be omitted according to \hyperref[lem:neutral_combination]{Lemma~\ref*{lem:neutral_combination}}.
	Likewise, any delta tensor can be removed with $\delta$-merging. To ensure that both index symbols of the $\delta$-tensor appear elsewhere in the expression, we can introduce all-ones vectors with \hyperref[lem:neutral_combination]{Lemma~\ref*{lem:neutral_combination}}.
\end{proof}

\section{Conclusions}
\label{sec:conclusions}

We have introduced a unified syntax and semantics for the einsum language and established semantic equivalence rules that provide a solid theoretical foundation for its practical use.
In particular, we proved that einsum expressions over arbitrary commutative semirings are commutative, associative, and distributive with respect to elementwise aggregation.
In addition, we developed general nesting and denesting rules to systematically restructure expressions, along with simplification rules that clarify when delta tensors or constant tensors can be introduced or removed without altering semantics.
Together, these results close the gap between the widespread practical adoption of einsum and its missing theoretical underpinnings,
creating a unified framework that supports principled reasoning about equivalence, differentiation, and optimization of tensor expressions across different computational settings.

\section*{Acknowledgements}
This work was supported by the Carl-Zeiss-Stiftung within the project `Interactive Inference'.

\nocite{*}
\bibliographystyle{fundam}
\bibliography{einsum}


\end{document}